%% file: arxiv_260219.tex
\documentclass[sigconf]{aamas}

\usepackage{balance} 

\usepackage{amsthm}
\usepackage{algorithm}
\usepackage{algorithmic}
\usepackage{mathtools}
\usepackage{color}
\usepackage{amsfonts}
\newcommand{\bs}{\boldsymbol}

\newcommand{\m}{{\rm m}}
\newcommand{\M}{{\rm M}}

\newcommand{\rmd}{{\rm d}}

\newcommand{\fst}{{\rm 1st}}
\newcommand{\snd}{{\rm 2nd}}

\newtheorem{theorem}{Theorem}
\newtheorem{lemma}{Lemma}

\newtheorem{definition}{Definition}

\usepackage{comment}
\usepackage{bm}
\usepackage{dsfont}

\setcopyright{ifaamas}
\acmConference[AAMAS '26]{Proc.\@ of the 25th International Conference
on Autonomous Agents and Multiagent Systems (AAMAS 2026)}{May 25 -- 29, 2026}
{Paphos, Cyprus}{C.~Amato, L.~Dennis, V.~Mascardi, J.~Thangarajah (eds.)}
\copyrightyear{2026}
\acmYear{2026}
\acmDOI{}
\acmPrice{}
\acmISBN{}

\acmSubmissionID{938}

\title[Time-Varyingness in Auction Breaks Revenue Equivalence]{Time-Varyingness in Auction Breaks Revenue Equivalence}
\subtitle{Extended Abstract}
\subtitlenote{This is the full version.}

\author{Yuma Fujimoto}
\affiliation{
  \institution{CyberAgent}
  \city{Tokyo}
  \country{Japan}}
\email{fujimoto.yuma1991@gmail.com}

\author{Kaito Ariu}
\affiliation{
  \institution{CyberAgent}
  \city{Tokyo}
  \country{Japan}}
\email{kaito_ariu@cyberagent.co.jp}

\author{Kenshi Abe}
\affiliation{
  \institution{CyberAgent}
  \city{Tokyo}
  \country{Japan}}
\email{abe_kenshi@cyberagent.co.jp}

\begin{abstract}
Auction is applied for trade with various mechanisms. A simple but practical question is which mechanism, typically first-price or second-price auctions, is preferred from the perspective of bidders or sellers. A celebrated answer is revenue equivalence, where each bidder's equilibrium payoff is proven to be independent of auction mechanisms (and a seller's revenue, too). In reality, however, auction environments like the value distribution of items would vary over time, and such equilibrium bidding cannot always be achieved. Indeed, bidders must continue to track their equilibrium bidding by learning in first-price auctions, but they can keep their equilibrium bidding in second-price auctions. This study discusses whether and how revenue equivalence is violated in the long run by comparing the time series of non-equilibrium bidding in first-price auctions with those of equilibrium bidding in second-price auctions. We characterize the value distribution by two parameters: its basis value, which means the lowest price to bid, and its value interval, which means the width of possible values. Surprisingly, our theorems and experiments find that revenue equivalence is broken by the correlation between the basis value and the value interval, uncovering a novel phenomenon that could occur in the real world.
\end{abstract}


\newcommand{\BibTeX}{\rm B\kern-.05em{\sc i\kern-.025em b}\kern-.08em\TeX}

\begin{document}


\pagestyle{fancy}
\fancyhead{}


\maketitle 


\section{Introduction}
Auction is a class of buying-selling systems including various mechanisms~\cite{vickrey1961counterspeculation}. Multiple buyers bid their prices for an item. The bidder who bids the highest price wins the item, but its payment depends on the auction mechanism. In advertising auctions, first-price and second-price mechanisms are typically adopted. The winner pays the highest price, i.e., the price bid by itself, in first-price auctions. On the other hand, it pays only the next highest price in second-price auctions. A captivating question is which mechanism is preferred from the perspectives of the bidders or the seller.

The revenue equivalence theorem is one powerful but counterintuitive theoretical result in auction theory~\cite{vickrey1961counterspeculation, riley1981optimal, myerson1981optimal}. It proves that in equilibrium, both the revenue of the seller and the payoff of bidders are independent of auction mechanisms. Revenue equivalence holds when the value of an item follows a symmetric, independent, and private distribution among bidders. It has been reported that revenue equivalence is broken down, for example, when the value distribution is asymmetric~\cite{griesmer1967toward, plum1992characterization, maskin2000asymmetric, cheng2006ranking} or interdependent~\cite{milgrom1982theory}. Discussing this revenue equivalence is crucial for designing advantageous auction mechanisms for sellers or bidders. As a related phenomenon, most supply-side platforms (sellers) have switched their mechanism from second-price to first-price~\cite{paes2020competitive, despotakis2021first, alcobendas2022adjustment, goke2022bidders}. Revenue equivalence is an important issue connected to real-world auctions.

\begin{figure*}[tb]
    \centering
    \includegraphics[width=0.8\hsize]{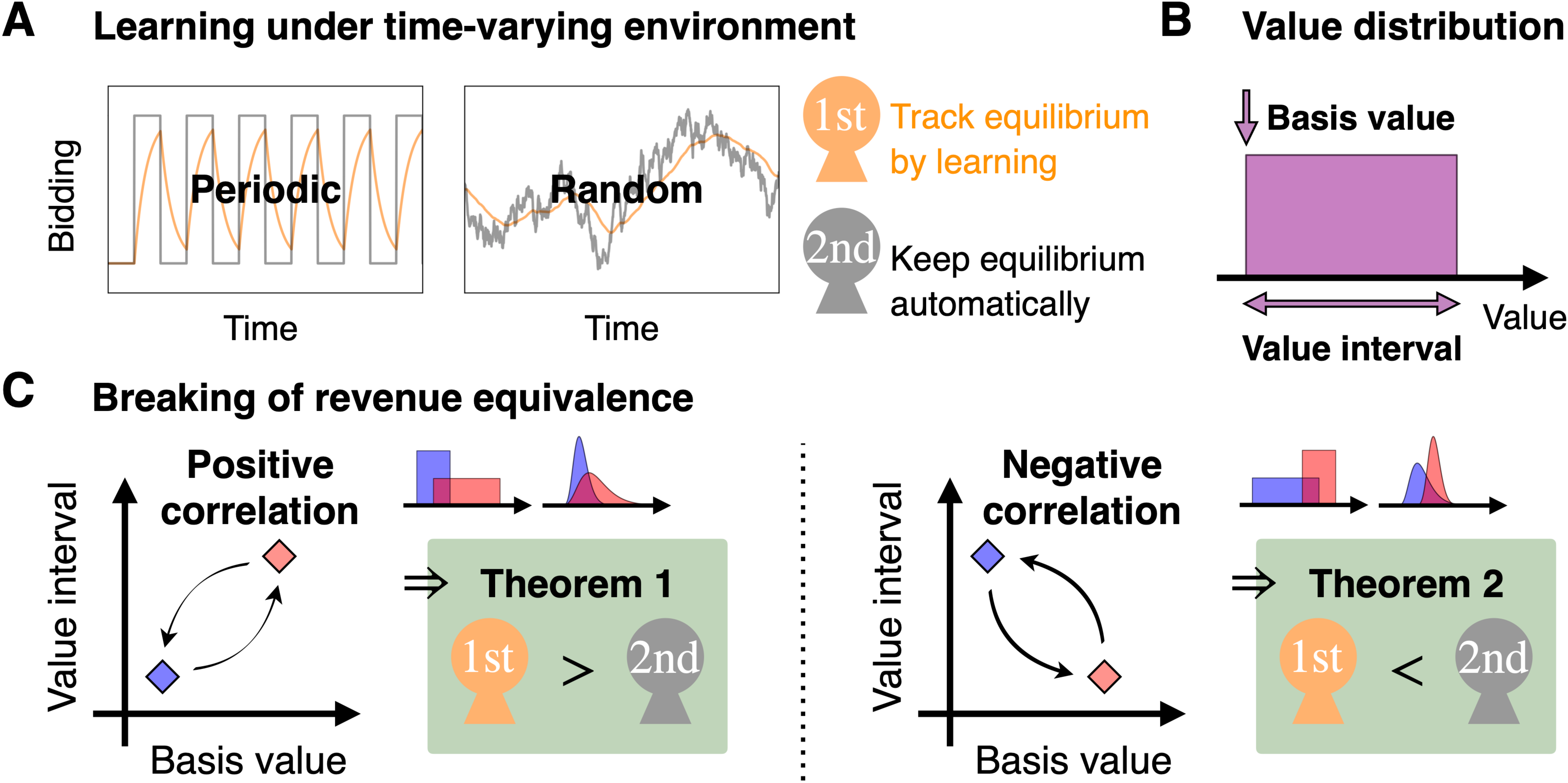}
    \caption{Overview of learning in time-varying auctions. (A). This study considers time-varying environments, such as periodic and random ones. How bidding changes over time under such time-varying environments is plotted for first-price (orange) and second-price (gray) auctions. In second-price auctions, truthful bidders automatically keep their optimal bidding regardless of the environmental change. In first-price auctions, however, bidders should track the optimal bidding by learning. (B). A value distribution (e.g., uniform distribution) is characterized by its basis value and its value interval. (C). This study demonstrates that the correlation between the basis value and the value interval determines how revenue equivalence is broken. The positive correlation results in bidders preferring first-price auctions (Thm.~\ref{thm_inequivalence_1st}), while the negative correlation results in bidders preferring second-price auctions (Thm.~\ref{thm_inequivalence_2nd}). The upper-right distributions are examples of uniform and log-normal distributions that provide positive/negative correlation.}
    \label{F01}
\end{figure*}

Although discussing whether the revenue equivalence holds or not is based on the equilibrium analysis, real-world auctions typically do not reach such an equilibrium. One primary factor is a time-varying environment, which continues to change the equilibrium over time (see Fig.~\ref{F01}-A). Indeed, daily and weekly cycles equipped with randomness, possibly caused by the rhythm of human life, are observed in bidding from empirical auction data~\cite{yuan2013real, yuan2014empirical, goke2022bidders}. The resistance to the time-varying equilibrium differs between the first-price and second-price auctions. Since bidders should bid the equilibrium payment in a first-price auction, they can only track the equilibrium by learning at best (see the orange parts in Fig.~\ref{F01}-A). On the other hand, in second-price auctions, they only have to bid the value they perceive (called truthful bidding), leading to the equilibrium payment. Hence, they can keep the equilibrium automatically (see the gray parts). Thus, when auction environments vary, one mechanism may be favorable to either bidders or sellers than the other. This study addresses the following question;
\begin{center}
\textit{Does dynamic value, even if symmetric, independent, and private, break revenue equivalence?}
\end{center}
If so, the further question arises.
\begin{center}
\textit{Which of the first-price and second-price auctions \\ are preferred by bidders (or sellers)?}
\end{center}
Answering these questions may explain which mechanism should be implemented in reality.

To these questions, the main contributions of this study are as follows.
\begin{itemize}
\item \textbf{We characterize the long-run behavior where bidders track time-varying equilibria.} We extend the symmetric Bayesian Nash equilibrium for a given value distribution and formulate the gradient descent-ascent dynamics where bidders learn unknown true parameters for a time-varying value distribution. To test revenue equivalence in a time-varying environment, we compare the long-run payoffs between first-price and second-price auctions under time-varying value distributions.
\item \textbf{We provide theoretical conditions under which revenue equivalence is broken for uniform value distributions.} A uniform distribution is characterized by two parameters, its basis value and its value interval (see Fig.~\ref{F01}-B), and the correlation of these parameters over time determines the revenue inequivalence (see Fig.~\ref{F01}-C). When these parameters correlate positively (Thm.~\ref{thm_inequivalence_1st}), bidders receive a greater payoff in first-price auctions. The inverse is also true (Thm.~\ref{thm_inequivalence_2nd}).
\item \textbf{We extract insights from our theory and experimentally demonstrate the revenue inequivalence.} Our experiments show that the revenue equivalence is broken in both directions in periodic and random environments, respectively. We also capture an intuition that for log-normal distributions, the correlation between the mean and the variance breaks revenue equivalence, and our further experiments verify this intuition.
\end{itemize}

\subsection{Continuous-Time Analysis}
This study adopts a continuous-time analysis approach for time-varying auctions. Auctions exist within the framework of game theory, where one's optimal strategy depends on the others' strategies, and thus learning in games frequently faces complex behaviors, such as non-convergence (cycling). The continuous-time analysis has succeeded in capturing these complex behaviors~\cite{sato2002chaos, bloembergen2015evolutionary, mertikopoulos2016learning, mertikopoulos2018cycles, bailey2019multi, fujimoto2025global, ota2025hamiltonian}. When the game environment varies over time, more complex behaviors occur and require deep understanding by the continuous-time analysis~\cite{fiez2021online, fujimoto2025synchronization}. Non-convergence behaviors are also seen in learning in an auction~\cite{kolumbus2022auctions, deng2022nash}, and the continuous-time analysis is often taken~\cite{paes2024complex}. In this study, the continuous-time analysis is helpful to capture similar non-convergence (equilibrium tracking) behavior and well approximates the behavior of learning with a small learning rate.

\section{Preliminary: Static Environment}
First, we introduce the classical setting, where auction environments, such as value distribution, are static.

\subsection{Setting}
Auction theory considers $n\in\mathbb{N}$ bidders, labeled as $i\in\{1,\cdots, n\}$, in general. An item is offered to the bidders every time, and its value is $v_{i}>0$ for bidder $i$. Each bidder independently determines the bidding $b_{i}$ for the item. The bidder who bids the highest price receives the item. However, its payment depends on the type of auction. This study considers two representative types: first-price and second-price auctions. In the former, the payment is the first price ($b^{\fst}=\max_{j}b_{j}=b_{i}$), i.e., the bidding of the winner. In the latter, the payment is the second price ($b^{\snd}=\max_{j\neq i}b_{j}$), i.e., the maximum bidding of those other than the winner. Thus, the payoffs of first-price and second-price auctions are
\begin{align}
    u_i^{\fst}(b_i|v_i)&=(v_i-b^{\fst})\mathds{1}[b_i=\max_{j}b_j],
    \label{payoff_1st}\\
    u_i^{\snd}(b_i|v_i)&=(v_i-b^{\snd})\mathds{1}[b_i=\max_{j}b_j].
    \label{payoff_2nd}
\end{align}
In auctions, the strategy of each bidder is bidding $b_i$ to the observed value $v_i$, i.e., the function of $b_i(v_i)$. This study considers that $v_i$ follows the same distribution $f(v_{i};\bs{\theta}^{*})$ on $[0,\infty)$ independently for each $i$ (called the setting of ``symmetric, independent, and private'' value~\cite{krishna2009auction}). Here, we assume that $f(v_{i};\bs{\theta}^{*})$ is characterized by finite parameters, denoted as $\bs{\theta}^{*}\in\mathbb{R}^{d}$ and that $F(v_{i};\bs{\theta}^{*})$, the cumulative distribution of $f(v_{i};\bs{\theta}^{*})$, is continuous for $v_{i}$ and $\bs{\theta}^{*}$.

\subsection{Equilibrium Analysis}
The equilibrium in such auctions has already been studied. The symmetric Bayesian Nash equilibrium in a first-price auction is $b_{i}=b^{*}$ for all $i$ such that
\begin{align}
    b^{*}(v)=v-\frac{1}{F(v;\bs{\theta}^{*})^{n-1}}\int_{0}^{v}F(z;\bs{\theta}^{*})^{n-1}\rmd z.
    \label{SBNE}
\end{align}
By the definition of $b^{\fst}$, this $b^{*}(v)$ is equal to the winner's payment. On the other hand, the equilibrium in a second-price auction is given by truthful bidding $b_{i}=b^{\rm TB}$ for all $i$ such that
\begin{align}
    b^{\rm TB}(v)=v,
    \label{TB}
\end{align}
meaning that each one bids the value it perceives. Interestingly, the revenue equivalence theorem (reviewed in Appendix~\ref{app_review}) shows that the winner's payment is equal between first-price and second-price auctions in its expectation.

\section{Time-Varying Value Distribution}
This study considers a situation where the value distribution varies over time. In detail, we assume that the parameters of the value distribution depend on time, i.e., $\bs{\theta}_{s}^{*}$ for $s\in\{0,1,\cdots\}$.

Now, imagine the behavior of bidders in first-price auctions. Since the equilibrium strategy in a first-price auction explicitly depends on the true parameters $\bs{\theta}_{s}^{*}$, it would be unrealistic to continue to know these true parameters. Instead, at best, bidders would guess that the parameters are $\bs{\theta}_{s}$ instead of the true one $\bs{\theta}_{s}^{*}$ and use the equilibrium bidding $b(v;\bs{\theta}_{s})$ based on $\bs{\theta}_{s}$, which is defined as
\begin{align}
    b(v;\bs{\theta}_{s})=v-\frac{1}{F(v;\bs{\theta}_{s})^{n-1}}\int_{0}^{v}F(z;\bs{\theta}_{s})^{n-1}\rmd z.
\end{align}
They would further learn $\bs{\theta}_{s}$ to track $\bs{\theta}_{s}^{*}$.

On the other hand, consider the behavior of bidders in second-price auctions. Recall that the equilibrium bidding, i.e., truthful bidding, is independent of the true parameters, and this truthful bidding is feasible even when the true parameters vary over time. Hence, bidders can continue to achieve the equilibrium payment $b(v;\bs{\theta}_{s}^{*})$.

In summary, the bidding behaviors of first-price and second-price auctions crucially differ in such time-varying value distributions as follows.
\begin{itemize}
    \item In a first-price auction, bidders pay $b(v;\bs{\theta}_{s})$ and learn $\bs{\theta}_{s}$ to track the true parameters $\bs{\theta}_{s}^{*}$.
    \item In a second-price auction, bidders can automatically keep the equilibrium payment $b(v;\bs{\theta}_{s}^{*})$.
\end{itemize}
To discuss how such different bidding behaviors break revenue equivalence, we propose a promising algorithm to learn $\bs{\theta}_{s}$ in first-price auctions and formulate the time-averaged payoffs both in first-price and second-price auctions.

\subsection{Learning Dynamics}
Let $w_{\bs{\theta}^{*}}(\bs{\theta}',\bs{\theta})$ denote the expected payoff of a focal bidder who uses the bidding $b(v;\bs{\theta}')$ while all the others use $b(v;\bs{\theta})$ under the true parameter $\bs{\theta}^{*}$. This expected payoff is described as
\begin{align}
    w_{\bs{\theta}^{*}}(\bs{\theta}',\bs{\theta})=\int_{0}^{\infty}(v-b(v;\bs{\theta}'))f(v;\bs{\theta}^{*})F(v';\bs{\theta}^{*})^{n-1}\rmd v.
\end{align}
Here, we defined $v'=v'(v,\bs{\theta},\bs{\theta}')$ such that $b(v';\bs{\theta})=b(v;\bs{\theta}')$, meaning that when the focal bidder observes the item value $v$ with parameter $\bs{\theta}'$, it bids as if it observes $v'$ with parameter $\bs{\theta}$.

Suppose that at time $s$, all the bidders estimate the parameter as $\bs{\theta}_{s}$ under the true parameter $\bs{\theta}^{*}_{s}$. We assume that each bidder observes the gradient of its expected payoff for its parameter (called the ``full-feedback'' setting), and each bidder can update its parameter by following the gradient descent-ascent as
\begin{align}
    \bs{\theta}_{s+1}=\bs{\theta}_{s}+\eta\left.\frac{\partial w_{\bs{\theta}^{*}_{s}}(\bs{\theta}',\bs{\theta}_{s})}{\partial\bs{\theta}'}\right|_{\bs{\theta}'=\bs{\theta}_{s}},
\end{align}
where $\eta\in\mathbb{R}$ is the learning rate.

For analysis, we introduce the continuum limit of the gradient descent-ascent dynamics for time $t\in[0,T]$. According to the relation of $t=s\eta$, we can introduce the continuous-time differentiation and obtain the continuum learning dynamics of
\begin{align}
    \dot{\bs{\theta}}(t)=\left.\frac{\partial w_{\bs{\theta}^{*}(t)}(\bs{\theta}',\bs{\theta}(t))}{\partial \bs{\theta}'}\right|_{\bs{\theta}'=\bs{\theta}(t)}.
    \label{dotx}
\end{align}
This well approximates the behavior for small $\eta$.

\paragraph{Remark on symmetry:} The learning dynamics (Eq.~\eqref{dotx}) assume that all bidders have the same parameters $\bs{\theta}(t)$. The symmetric but non-equilibrium setting is a natural extension of the symmetric Bayesian Nash equilibrium, which is assumed in the classical auction theory. Furthermore, when all bidders are under the same time-varying environment, their parameters are expected to gradually synchronize through learning. At least once these parameters synchronize, they continue to synchronize as long as they follow the dynamics. Furthermore, the dynamics are guaranteed to converge to the equilibrium for uniform distributions, as shown later. The final remark is that a similar idealization is adopted in the context of evolutionary biology as adaptive dynamics~\cite{hofbauer1990adaptive, dieckmann1996dynamical, nowak2004evolutionary}.

\begin{figure*}[tb]
    \centering
    \includegraphics[width=0.8\hsize]{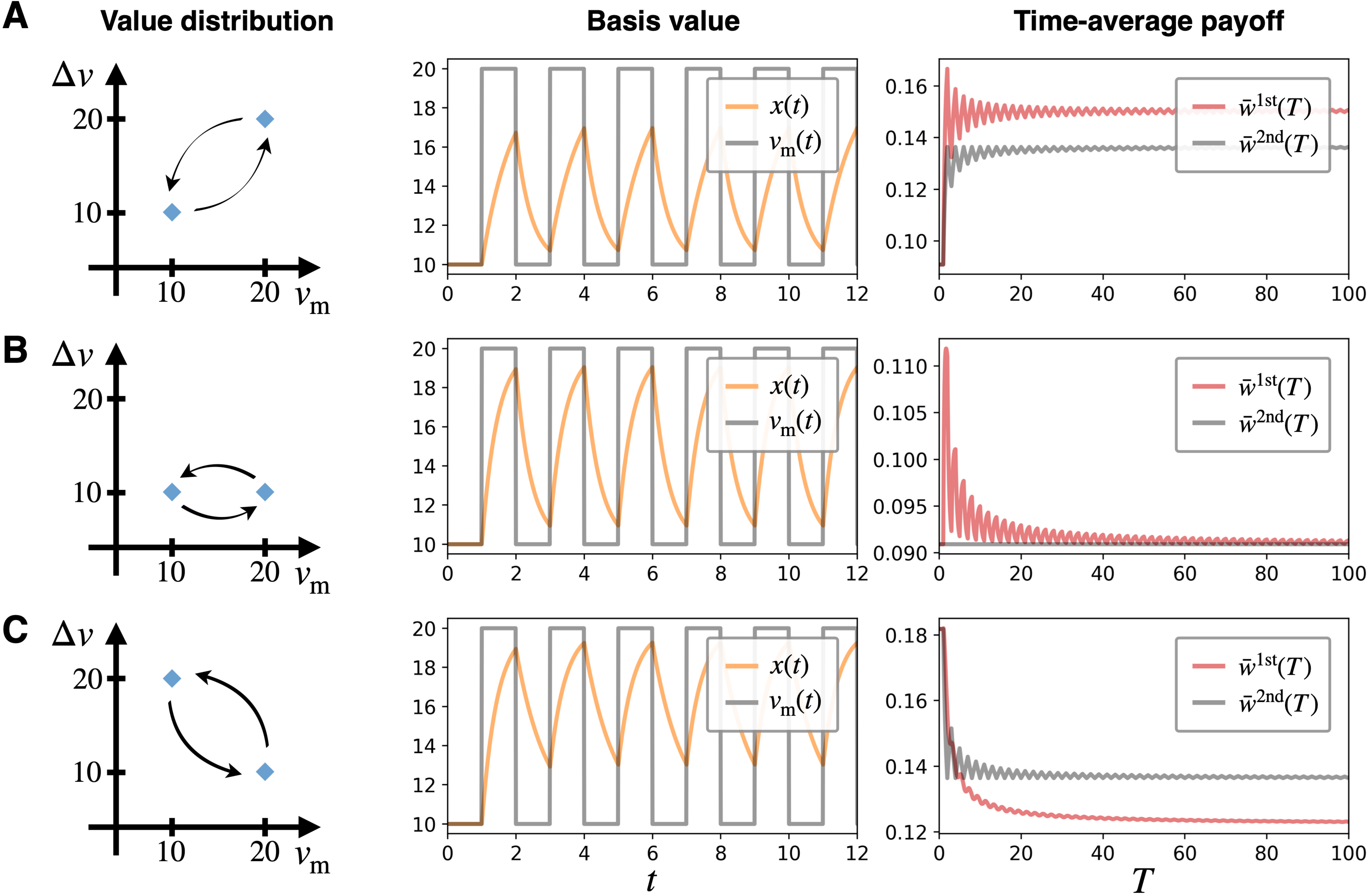}
    \caption{The experiments for periodic uniform distributions. The left panels show how the uniform distribution switches between two states. The center panels show the time series of the learned basis value ($x(t)$: orange) and the true basis value ($v_{\m}(t)$: gray). The right panels show the time series of the time-average payoffs in the first-price auction ($\bar{w}^{\fst}(T)$: red) and in the second-price auction ($\bar{w}^{\snd}(T)$: gray) within the range of $0\le t\le T$. In all the experiments, we set the population as $n=10$, the Runge-Kutta fourth-order method of Eq.~\eqref{dotx} with the step size of $10^{-3}$ and accelerated $2\times 10^{3}$. The case of A considers $\bs{\theta}^{*}(t)\in\{(10,20),(20,40)\}$, B does $\bs{\theta}^{*}(t)\in\{(10,20),(20,30)\}$, and C does $\bs{\theta}^{*}(t)\in\{(10,30),(20,30)\}$.}
    \label{F02}
\end{figure*}

\begin{figure*}[tb]
    \centering
    \includegraphics[width=0.8\hsize]{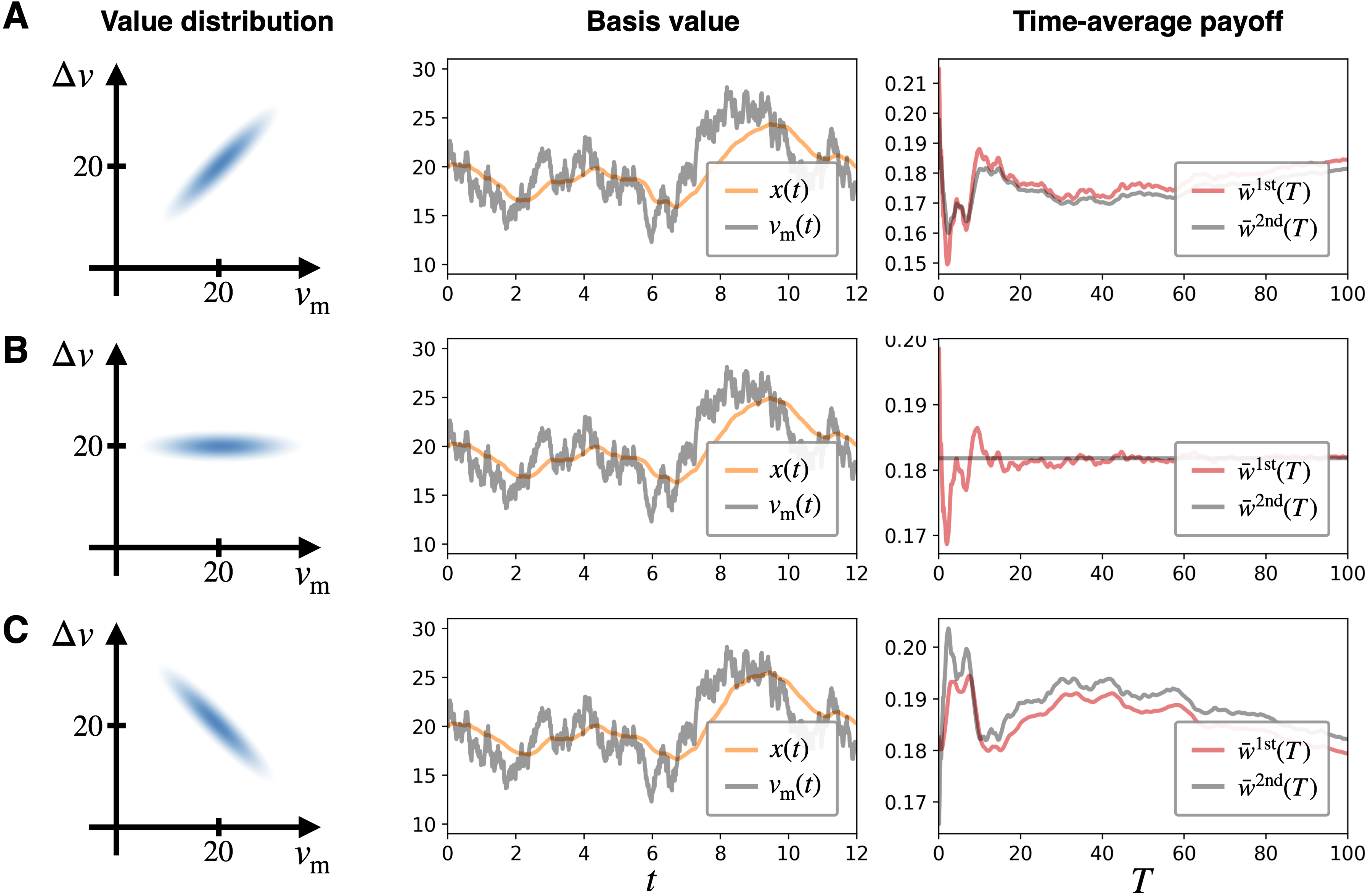}
    \caption{The experiments for random uniform distributions. The left panels show possible uniform distributions generated by Eqs.~\eqref{Langevin_m} and~\eqref{Langevin_M}. The meaning of the center and right panels is the same as Fig.~\ref{F02}. The methods and parameters for all the experiments, too. In all the cases of A-C, we set $(\bar{v}_{\m},\bar{v}_{\M})=(20,40)$. The case of A considers $(a_{\m},a_{\M})=(5,10)$ in Eqs.~\eqref{Langevin_m} and~\eqref{Langevin_M}, B does $(a_{\m},a_{\M})=(5,5)$, and C does $(a_{\m},a_{\M})=(5,0)$.}
    \label{F03}
\end{figure*}

\subsection{Long-Run Payoff}
This study evaluates the payoff of bidders in the long run. Each time $t$, the expected payoffs in first-price and second-price auctions are described as
\begin{align}
    w_{\bs{\theta}^{*}(t)}^{\fst}(\bs{\theta}(t))&=w_{\bs{\theta}^{*}(t)}(\bs{\theta}(t),\bs{\theta}(t)), \\
    w_{\bs{\theta}^{*}(t)}^{\snd}&=w_{\bs{\theta}^{*}(t)}^{\fst}(\bs{\theta}^{*}(t)).
\end{align}
Furthermore, the time-averaged payoffs for $0\le t\le T$ are
\begin{align}
    \bar{w}^{\fst}(T)&=\frac{1}{T}\int_{0}^{T}w_{\bs{\theta}^{*}(t)}^{\fst}(\bs{\theta}(t))\rmd t, \\
    \bar{w}^{\snd}(T)&=\frac{1}{T}\int_{0}^{T}w_{\bs{\theta}^{*}(t)}^{\snd}\rmd t.
\end{align}

By comparing $\bar{w}^{\fst}(\infty)$ and $\bar{w}^{\snd}(\infty)$, we discuss which first-price or second-price auction is preferable for bidders and sellers in the long run of time-varying auctions.
\begin{itemize}
    \item When $\bar{w}^{\fst}(\infty)>\bar{w}^{\snd}(\infty)$, revenue equivalence is broken, and the bidders prefer to participate in first-price auctions. Since the bidders' payoffs and the seller's revenue are zero-sum, the seller prefers to hold second-price auctions.
    \item When $\bar{w}^{\fst}(\infty)<\bar{w}^{\snd}(\infty)$, revenue equivalence is broken, but the preference of the bidders and the seller is reversed.
    \item When $\bar{w}^{\fst}(\infty)=\bar{w}^{\snd}(\infty)$, revenue equivalence holds. Neither the bidders nor the seller has any preference for the auction type. This means that the classical revenue equivalence can be extended to time-varying auctions.
\end{itemize}

\section{Theoretical Results}
Let us capture theoretical insight into the learning dynamics (Eq.~\eqref{dotx}). To this end, we focus on a class of value distributions, uniform distributions.

\begin{definition}[Uniform distribution]
Uniform distributions assume the parameters $\bs{\theta}^{*}(t)=(v_{\m}(t),v_{\M}(t))$ and are defined as
\begin{align}
    f(v;\bs{\theta}^{*}(t))=\frac{1}{v_{\M}(t)-v_{\m}(t)}\mathds{1}[v_{\m}(t)\le v\le v_{\M}(t)].
\end{align}
We also define the estimated parameters as $\bs{\theta}(t)=(x(t),y(t))$.
\end{definition}

Under this uniform distribution, each bidder observes a random value from $v_{\m}(t)$ to $v_{\M}(t)$. For convenience, we define the interval of the value distribution (simply, ``value interval'') as
\begin{align}
    \Delta v(t):=v_{\M}(t)-v_{\m}(t).
\end{align}

For uniform distribution, we obtain the bidding function as
\begin{align}
    b(v;\bs{\theta}(t))&=\underbrace{\frac{n-1}{n}}_{=:\alpha}(v-x(t))+x(t).
\end{align}
Here, note that $b(v;\bs{\theta}^{*}(t))\ge v_{\m}(t)$ holds for all $v_{\m}(t)\le v\le v_{\M}(t)$. Thus, $v_{\m}(t)$ is the basis in bidding (simply, ``basis value''). We also obtain the learning dynamics and the expected payoff as
\begin{align}
    \dot{x}(t)&=-\frac{x(t)-v_{\m}(t)}{n(n-1)\Delta v(t)},\quad \dot{y}(t)=0, \\
    w_{\bs{\theta}^{*}(t)}^{\fst}(\bs{\theta}(t))&=w_{\bs{\theta}^{*}(t)}^{\snd}-\frac{x(t)-v_{\m}(t)}{n^2},\quad w_{\bs{\theta}^{*}(t)}^{\snd}=\frac{\Delta v(t)}{n(n+1)},
\end{align}
(see Appendix~\ref{app_uniform} for the detailed calculations).

We now see that when the value distribution is time-invariant, i.e., $v_{\m}(t)=v_{\m}$ and $v_{\M}(t)=v_{\M}$, the gradient descent-ascent provides a reasonable performance. Obviously, $\dot{x}(t)<0$ holds for $x(t)>v_{\m}$, while $\dot{x}(t)>0$ for $x(t)<v_{\m}$. Hence, the convergence to the equilibrium ($\lim _{t\to \infty}x(t)=v_{\m}$) is guaranteed. The payoff in a first-price auction $w_{\bs{\theta}^{*}}^{\fst}(\bs{\theta}(t))$ also converges to that in a second-price one $w_{\bs{\theta}^{*}}^{\snd}$, immediately leading to $\bar{w}^{\fst}(\infty)=\bar{w}^{\snd}(\infty)$.

\subsection{Revenue Inequivalence}
This section discusses whether and how revenue equivalence is broken, based on the long-run payoff. First, when the basis value positively correlates with the value interval, bidders prefer first-price auctions in any time-varying value distributions (see Appendix~\ref{app_proof} for the full proof).

\begin{theorem}[Revenue inequivalence by positive correlation]
\label{thm_inequivalence_1st}
Suppose that the true parameter can take any $K$-states, i.e., $\bs{\theta}^{*}(t)\in\{\bs{\theta}^{(1)},\cdots,\bs{\theta}^{(K)}\}$ and that $v_{\m}^{(k)}<v_{\m}^{(k')} \Rightarrow \Delta v^{(k)}<\Delta v^{(k')}$ holds for all $k,k'\in\{1,\cdots,K\}$, then $\bar{w}^{\fst}(\infty)>\bar{w}^{\snd}(\infty)$ holds.
\end{theorem}

\textsc{Proof Sketch.} We assume time series of $x$ for $t\in[0,T]=:\Omega$ and discuss the time series of $\bs{\theta}^{*}$ which give the lower bound of $\bar{w}^{\fst}(\infty)-\bar{w}^{\snd}(\infty)$. We first divide the whole time $\Omega$ into $\Omega^{+}$ where $\dot{x}\ge 0$ holds, and $\Omega^{-}$ where $\dot{x}<0$ holds. Let $\tilde{k}(x)$ be $k$ such that $v_{\m}^{(k)}\le x< v_{\m}^{(k+1)}$ holds. We show that under $\Omega^{+}$, $w_{\bs{\theta}^{*}(t)}^{\fst}(\bs{\theta})-w_{\bs{\theta}^{*}(t)}^{\snd}$ takes minimum when $k=\tilde{k}(x)+1$, while under $\Omega^{-}$, it does when $k=\tilde{k}(x)$. Thus, in the long run, the lower bound of $\bar{w}^{\fst}(\infty)-\bar{w}^{\snd}(\infty)$ is given by the product of the path length of $x$ and the minimum of $\Delta v^{(k+1)}-\Delta v^{(k)}(>0)$ for all $k$. This means $\bar{w}^{\fst}(\infty)>\bar{w}^{\snd}(\infty)$.\qed \\

\begin{figure*}[tb]
    \centering
    \includegraphics[width=0.8\hsize]{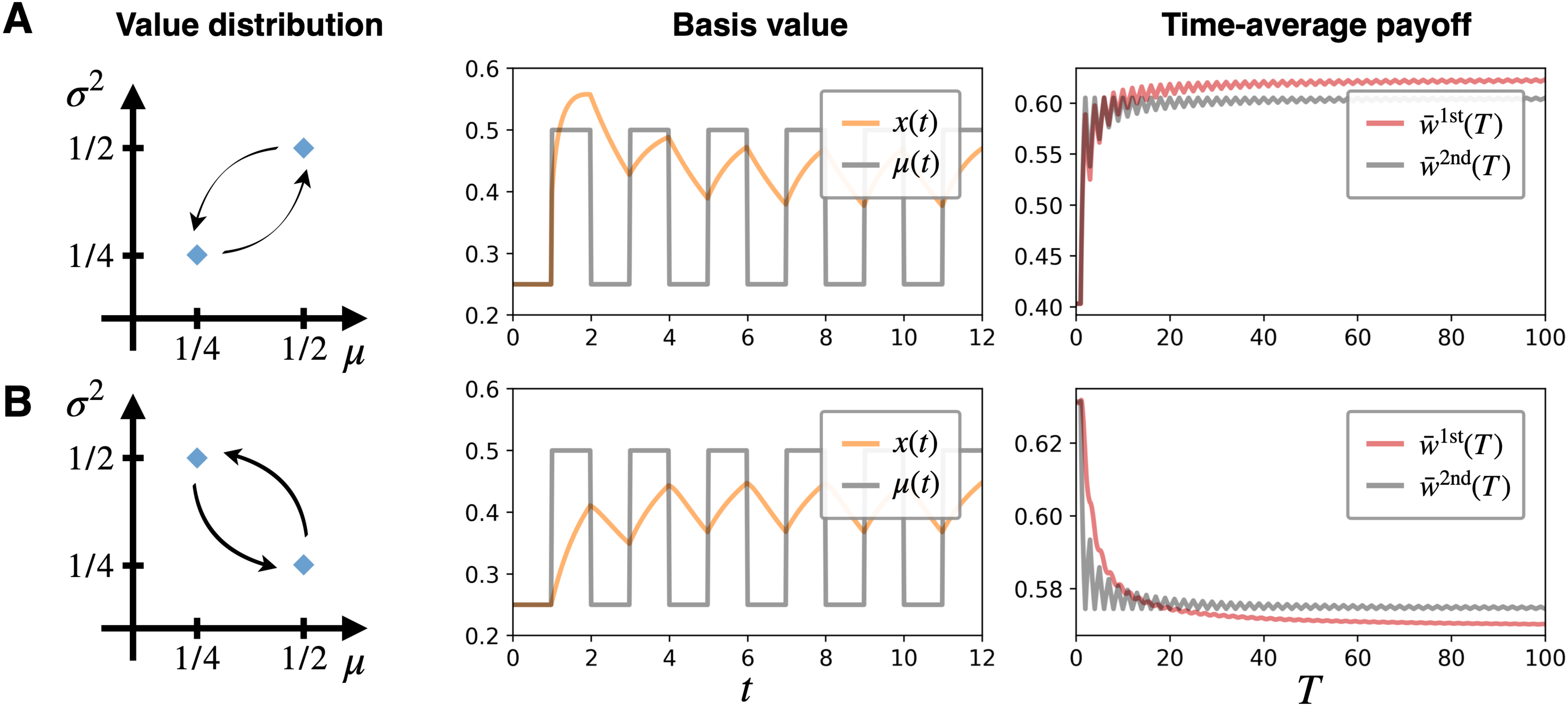}
    \caption{The experiments for periodic log-normal distributions. The left panels show how the log-normal distribution switches between two states. The meaning of the center and right panels is the same as Fig.~\ref{F02}. To simulate the dynamics of Eq.~\eqref{dotx}, we numerically calculate all the integrals by discretizing the space of $0\le v\le 20$ with $400$ meshes. The methods and parameters for all the experiments, too. The case of A considers $\bs{\theta}^{*}(t)\in\{(1/4,1/4), (1/2,1/2)\}$, B does $\bs{\theta}^{*}(t)\in\{(1/4,1/2),(1/2,1/4)\}$.}
    \label{F04}
\end{figure*}

On the other hand, when the basis value negatively correlates with the value interval, bidders prefer second-price auctions.

\begin{theorem}[Revenue inequivalence by negative correlation]
\label{thm_inequivalence_2nd}
Suppose that the true parameter can take any $K$-states, i.e., $\bs{\theta}^{*}(t)\in\{\bs{\theta}^{(1)},\cdots,\bs{\theta}^{(K)}\}$ and that $v_{\m}^{(k)}<v_{\m}^{(k')} \Rightarrow \Delta v^{(k)}>\Delta v^{(k')}$ holds for all $k,k'\in\{1,\cdots,K\}$, then $\bar{w}^{\fst}(\infty)<\bar{w}^{\snd}(\infty)$ holds.
\end{theorem}

\begin{proof}
This theorem is proved by reversing the direction of all the inequalities in the proof of Thm.~\ref{thm_inequivalence_1st}. Thus, $\bar{w}^{\fst}(\infty)-\bar{w}^{\snd}(\infty)$ is upper-bounded by the product of the path length of $x$ and the maximum of $\Delta v^{(k+1)}-\Delta v^{(k)}(<0)$ for all $k$. This means $\bar{w}^{\fst}(\infty)<\bar{w}^{\snd}(\infty)$.
\end{proof}

Finally, when the basis value does not correlate with the value interval, the revenue equivalence theorem is maintained.

\begin{theorem}[Revenue equivalence by no correlation]
\label{thm_equivalence}
Suppose that $\Delta v(t)$ is constant for $t\in[0,T]$, then $\bar{w}^{\fst}(\infty)=\bar{w}^{\snd}(\infty)$ holds.
\end{theorem}
\begin{proof}
We denote the interval of the value distribution as $\Delta v(t)=\Delta v$ for all $t$, and the difference of the payoffs is evaluated as
\begin{align}
    \bar{w}^{\fst}(\infty)-\bar{w}^{\snd}(\infty)&=\alpha\lim_{T\to\infty}\frac{1}{T}\int_{0}^{T}\Delta v(t) \dot{x}(t)\rmd t
    \nonumber\\
    &=\alpha\Delta v \lim_{T\to\infty}\frac{1}{T}\int_{0}^{T}\dot{x}(t)\rmd t
    \nonumber\\
    &=0,
\end{align}
which immediately leads to long-run revenue equivalence, i.e., $\bar{w}^{\fst}(\infty)=\bar{w}^{\snd}(\infty)$.
\end{proof}

\subsection{Key Insight from Theorems}
Let us capture a key insight of why the positive correlation between the basis value $v_{\m}(t)$ and the value interval $\Delta v(t)$ leads to revenue inequivalence $\bar{w}^{\fst}(\infty)>\bar{w}^{\snd}(\infty)$. A key equation is as follows.
\begin{align}
    w_{\bs{\theta}^{*}(t)}^{\fst}(\bs{\theta}(t))-w_{\bs{\theta}^{*}(t)}^{\snd}\propto -(x(t)-v_{\m}(t))\propto\Delta v(t)\dot{x}(t).
\end{align}
The first proportional symbol shows that the difference of payoffs between first-price and second-price auctions is proportional to the non-equilibrium degree $x(t)-v_{\m}(t)$. The second proportional symbol also shows that the learning speed $\dot{x}(t)$ is proportional to the non-equilibrium degree $x(t)-v_{\m}(t)$, but it is inversely proportional to the value interval $\Delta v(t)$. Here, the latter is because it is difficult to learn bidding when there is dispersion in the observed prices. Furthermore, because the estimation $x(t)$ tracks the basis value $v_{\m}(t)$, $\dot{x}(t)$ tends to be positive when $v_{\m}(t)$ is large, while $\dot{x}(t)$ tends to be negative when $v_{\m}(t)$ is small. Hence, when the basis value positively correlates with the value interval, $\bar{w}^{\fst}(\infty)>\bar{w}^{\snd}(\infty)$ holds.

This key insight is applicable to situations beyond what the theorems suppose. It has been explained by the correlation between the basis value and the value interval, which are the parameters in uniform distributions. However, we can find the parameters that roughly correspond to the basis value and the value interval in other value distributions. This is because the bidding $b(v;\bs{\theta})$ is continuous for the cumulative distribution $F$ and thus is not so sensitive to the details of the value distribution $f$. For example, consider log-normal distributions, which are frequently applied in auctions~\cite{edelman2006optimal, xiao2009optimal, ostrovsky2011reserve}.

\begin{definition}[Log-normal distribution]
Log-normal distribution is defined by $\bs{\theta}^{*}(t)=(\mu(t),\sigma(t)^{2})$ as
\begin{align}
    f(v;\bs{\theta}^{*}(t))=\frac{1}{\sqrt{2\pi\sigma(t)^{2}}v}\exp\left(-\frac{(\log v-\mu(t))^{2}}{2\sigma(t)^{2}}\right).
\end{align}
\end{definition}

In this log-normal distribution, its mean value $\mu(t)$ roughly corresponds to the basis value because as it increases, the whole value distribution shifts to the right. On the other hand, its variance $\sigma(t)^{2}$ roughly corresponds to the value interval because as it increases, the range of possible values expands. Later, we will numerically demonstrate that the positive correlation between $\mu(t)$ and $\sigma(t)^{2}$ results in $\bar{w}^{\fst}(\infty)>\bar{w}^{\snd}(\infty)$ (Thm.~\ref{thm_inequivalence_1st}), while the negative correlation results in $\bar{w}^{\fst}(\infty)<\bar{w}^{\snd}(\infty)$ (Thm.~\ref{thm_inequivalence_2nd}).

\section{Experiment}
We now see the breaking of revenue equivalence by numerical calculation for various possible situations. First, Fig.~\ref{F02} considers that a uniform value distribution varies periodically. Next, Fig.~\ref{F03} considers that a uniform value distribution varies randomly, following a Langevin equation. Last, Fig.~\ref{F04} considers that a log-normal distribution varies periodically.

\subsection{Periodic Environment}
Fig.~\ref{F02}-A, B, and C consider the positive, no, and negative correlation between the basis value $v_{\m}(t)$ and the value interval $\Delta v(t)$, respectively. In all the cases, we see that the learned basis value $x$ in first-price auctions successfully tracks the true one $v_{\m}(t)$. However, the paths of $x$ differ among these cases depending on how $\Delta v(t)$ is paired with $v_{\m}(t)$. In the case of A, because $x$ is smaller than $v_{\m}(t)$ in average, revenue inequivalence $\bar{w}^{\fst}(T)>\bar{w}^{\snd}(T)$ holds as time $T$ passes (showing Thm.~\ref{thm_inequivalence_1st}). In B, revenue equivalence $\bar{w}^{\fst}(T)=\bar{w}^{\snd}(T)$ holds (showing Thm.~\ref{thm_equivalence}). In C, revenue inequivalence $\bar{w}^{\fst}(T)<\bar{w}^{\snd}(T)$ holds in the reverse inequality (showing Thm.~\ref{thm_inequivalence_2nd}).


\subsection{Random Environment}
We now consider a situation where the true parameters in uniform distributions $(v_{\m},v_{\M})$ vary randomly following the Langevin equations;
\begin{align}
    \dot{v}_{\m}(t)&=-(v_{\m}(t)-\bar{v}_{\m})+a_{\m}\eta(t),
    \label{Langevin_m}\\
    \dot{v}_{\M}(t)&=-(v_{\M}(t)-\bar{v}_{\M})+a_{\M}\eta(t).
    \label{Langevin_M}
\end{align}
Here, note that both $\dot{v}_{\m}$ and $\dot{v}_{\M}$ has the same noise term ($\left<\eta(t)\eta(t')\right>=\delta(t-t')$) but with different intensities ($a_{\m}>0$ and $a_{\M}>0$). The first terms in Eqs.~\eqref{Langevin_m} and~\eqref{Langevin_M} represent the restoring force to the original point of $(\bar{v}_{\m},\bar{v}_{\M})$.

Fig.~\ref{F03}-A, B, and C consider the positive ($a^{\m}<a^{\M}$), no ($a^{\m}=a^{\M}$), and negative ($a^{\m}>a^{\M}$) correlation between the basis value $v_{\m}$ and the value interval $\Delta v$. We see that the correlation between $v_{\m}$ and $\Delta v$ triggers the breaking of revenue inequivalence; $\bar{w}^{\fst}(T)>\bar{w}^{\snd}(T)$ holds as time $T$ passes in A, $\bar{w}^{\fst}(T)=\bar{w}^{\snd}(T)$ holds in B, and $\bar{w}^{\fst}(T)<\bar{w}^{\snd}(T)$ holds in C.


\subsection{Log-Normal Distribution}
To more broadly confirm the fact that the correlation between the basis value and the value distribution leads to revenue inequivalence, let us consider another class of value distributions, the log-normal one. In a log-normal distribution, its mean value $\mu$ roughly corresponds to the basis value, while its variance $\sigma^{2}$ to the value interval.

Fig.~\ref{F04}-A shows the positive correlation between $\mu$ and $\sigma^{2}$ results in revenue inequivalence $\bar{w}^{\fst}(T)>\bar{w}^{\snd}(T)$ as time $T$ passes. On the other hand, Fig.~\ref{F04}-B shows that the negative correlation between $\mu$ and $\sigma^{2}$ results in the reversed inequality $\bar{w}^{\fst}(T)<\bar{w}^{\snd}(T)$. These mean that the theoretical insight from Thms.~\ref{thm_inequivalence_1st} and~\ref{thm_inequivalence_2nd} is robust to the difference in the shape of value distributions.


\section{Conclusion}
This study extended the classical setting of fixed value distribution to a time-varying distribution. To discuss whether and how revenue equivalence is broken by time-varying value distribution, we formulate and compute the long-run bidding behavior, which differs between first-price and second-price auctions. By focusing on the class of uniform distributions, we obtained the theorems that support a key insight that the correlation between the basis value and the value interval triggers the breaking of revenue equivalence. Our experiments demonstrate that this insight is broadly applied to possible situations, such as periodic or random changes in uniform and log-normal distributions.

To pursue the principle of revenue inequivalence, this study introduced several idealizations. One future direction is to bring our setting closer to real-world auctions. First, it would be meaningful to analyze discrete-time dynamics and evaluate the effect of a finite learning rate on revenue equivalence. Also, it would be challenging to combine our time-varying setting with asymmetric or interdependent value distribution. Based on empirical auction data, we might specify how value distribution varies over time and how it affects revenue. The finding of this paper, i.e., the breaking of revenue equivalence due to time-varyingness, will provide a theoretical basis for this future direction.

\input{arxiv_260219.bbl}

\clearpage

\appendix
\onecolumn

\renewcommand{\theequation}{A\arabic{equation}}
\setcounter{equation}{0}
\renewcommand{\figurename}{Figure A}
\setcounter{figure}{0}
\setcounter{secnumdepth}{1}

\begin{center}
{\LARGE\bf Appendix}
\end{center}

\section{Review: Revenue Equivalence Theorem} \label{app_review}
We consider that all the bidders use the same bidding strategy $b(v)$, which is a monotonically increasing function. Here, imagine that one bidder deviates from them and bids $b(v')$ when the bidder is offered the item of value $v$. In other words, when the item has the value of $v$, the bidder behaves as if the item has the value of $v'$. We denote the payoff of the bidder as $u(v',v)$, which is formulated as
\begin{align}
    u(v',v)=vF(v')^{n-1}-p(v').
\end{align}
Here, we defined the expected payment as $P(v')$. Now, the gradient of $u(v',v)$ for $v'$ is
\begin{align}
    \frac{\partial u(v',v)}{\partial v'}=v(n-1)f(v')F(v')^{n-2}-p(v').
\end{align}
If $b(v)$ is the optimal strategy, $u(v',v)$ satisfies the extreme condition in $v'=v$ as
\begin{align}
    &\left.\frac{\partial u(v',v)}{\partial v'}\right|_{v'=v}=0 \\
    &\Leftrightarrow v(n-1)f(v)F(v)^{n-2}-p(v)=0 \\
    &\Leftrightarrow p(v)=v(n-1)f(v)F(v)^{n-2}.
\end{align}
Considering $P(0)=0$, we obtain
\begin{align}
    P(v)&=P(0)+\int_{0}^{v}p(v')\rmd v' \\
    &=\int_{0}^{v}v'(n-1)f(v')F(v')^{n-2}\rmd v' \\
    &=vF(v)^{n-1}-\int_{0}^{v}F(v')^{n-1}\rmd v'.
\end{align}
Here, recall that the bidder offered the item of value $v$ wins with the probability of $F(v)^{n-1}$. Under the condition that the bidder wins, the conditional expected payment is equal to the equilibrium bidding in the first-price auction $b^{*}(v)=P(v)/F(v)^{n-1}$ as
\begin{align}
    b^{*}(v)=v-\frac{1}{F(v)^{n-1}}\int_{0}^{v}F(v')^{n-1}\rmd v'.
\end{align}

\section{Detailed Calculation for General Distribution} \label{app_general}
In this section, we omit the dependence on $t$ and denote $\bs{\theta}(t)$ and $\bs{\theta}^{*}(t)$ as $\bs{\theta}$ and $\bs{\theta}^{*}$ for simplicity. We begin with the definition of the bidding function as
\begin{align}
    b(v;\bs{\theta})=v-\frac{1}{F(v;\bs{\theta})^{n-1}}\int_{0}^{v}F(z;\bs{\theta})^{n-1}\rmd z.
\end{align}

Suppose that a focal bidder uses $b(v;\bs{\theta}')$ when the others use $b(v;\bs{\theta})$. We define $v'=v'(v,\bs{\theta},\bs{\theta}')$ such that $b(v';\bs{\theta})=b(v;\bs{\theta}')$. This means that when the focal bidder observes the item value $v$ with parameter $\bs{\theta}'$, it bids as if it observes $v'$ with parameter $\bs{\theta}$. Now the focal bidder's payoff is described as
\begin{align}
    w_{\bs{\theta}^{*}}(\bs{\theta}',\bs{\theta})=\int_{0}^{\infty}(v-b(v;\bs{\theta}'))f(v;\bs{\theta}^{*})F(v';\bs{\theta}^{*})^{n-1}\rmd v,
\end{align}
where $(v-b(v;\bs{\theta}'))$ is the difference between the item value $v$ and the payment $b(v;\bs{\theta}')$. $f(v;\bs{\theta}^{*})$ is the probability that to the focal bidder observes value $v$. $F(v';\bs{\theta}^{*})^{n-1}$ is the probability that all the others observe lower values than $v'$.

By definition, $v'(v,\bs{\theta},\bs{\theta})=v$ holds. Furthermore, the relation between $\bs{\theta}'$ and $v'(v,\bs{\theta},\bs{\theta}')$ around $\bs{\theta}'=\bs{\theta}$ is obtained by the chain rule as
\begin{align}
    \frac{\partial b(v;\bs{\theta})}{\partial v}\left.\frac{\partial v'}{\partial\bs{\theta}'}\right|_{\bs{\theta}'=\bs{\theta}}=\frac{\partial b(v;\bs{\theta})}{\partial\bs{\theta}}.
\end{align}

Now, the payoff gradient for the parameters is calculated as
\begin{align}
    \left.\frac{\partial w_{\bs{\theta}^{*}}(\bs{\theta}',\bs{\theta})}{\partial \bs{\theta}'}\right|_{\bs{\theta}'=\bs{\theta}}=& \int_{0}^{\infty}-\frac{\partial b(v;\bs{\theta})}{\partial\bs{\theta}}f(v;\bs{\theta}^{*})F(v;\bs{\theta}^{*})^{n-1}+(v-b(v;\bs{\theta}))f(v;\bs{\theta}^{*})\frac{\partial F(v;\bs{\theta}^{*})^{n-1}}{\partial v}\left.\frac{\partial v'}{\partial \bs{\theta}'}\right|_{\bs{\theta}'=\bs{\theta}}\rmd v \\
    =&\ -\int_{0}^{\infty}\left(\frac{f(v;\bs{\theta})}{F(v;\bs{\theta})}\right)^{-1}\left(\frac{f(v;\bs{\theta})}{F(v;\bs{\theta})}-\frac{f(v;\bs{\theta}^{*})}{F(v;\bs{\theta}^{*})}\right)f(v;\bs{\theta}^{*})F(v;\bs{\theta}^{*})^{n-1}\frac{\partial b(v;\bs{\theta})}{\partial\bs{\theta}}\rmd v.
\end{align}
In the final line, we used
\begin{align}
    \left.\frac{\partial v'}{\partial\bs{\theta}'}\right|_{\bs{\theta}'=\bs{\theta}}&=\left(\frac{\partial b(v;\bs{\theta})}{\partial v}\right)^{-1}\frac{\partial b(v;\bs{\theta})}{\partial\bs{\theta}} \\
    &=\left((n-1)\frac{f(v;\bs{\theta})}{F(v;\bs{\theta})}(v-b(v;\bs{\theta}))\right)^{-1}\frac{\partial b(v;\bs{\theta})}{\partial\bs{\theta}}, \\
    \frac{\partial b(v;\bs{\theta})}{\partial v}&=1+(n-1)\frac{f(v;\bs{\theta})}{F(v;\bs{\theta})}\frac{1}{F(v;\bs{\theta})^{n-1}}\int_{0}^{v}F(z;\bs{\theta})^{n-1}\rmd z-1 \\
    &=(n-1)\frac{f(v;\bs{\theta})}{F(v;\bs{\theta})}(v-b(v;\bs{\theta})).
\end{align}

We describe the expected payoff as
\begin{align}
    w^{\fst}_{\bs{\theta}^{*}}(\bs{\theta})=w_{\bs{\theta}^{*}}(\bs{\theta},\bs{\theta})&=\int_{0}^{\infty}(v-b(v;\bs{\theta}))f(v;\bs{\theta}^{*})F(v;\bs{\theta}^{*})^{n-1}\rmd v, \\
    w^{\snd}_{\bs{\theta}^{*}}=w_{\bs{\theta}^{*}}^{\fst}(\bs{\theta}^{*})&=\int_{0}^{\infty}(v-b(v;\bs{\theta}^{*}))f(v;\bs{\theta}^{*})F(v;\bs{\theta}^{*})^{n-1}\rmd v \\
    &=\int_{0}^{\infty}\int_{0}^{v}F(v';\bs{\theta}^{*})^{n-1}\rmd v'f(v;\bs{\theta}^{*})\rmd v \\
    &=\left[F(v;\bs{\theta}^{*})\int_{0}^{v}F(v';\bs{\theta}^{*})^{n-1}\rmd v'\right]_{0}^{\infty}-\int_{0}^{\infty}F(v;\bs{\theta}^{*})^{n}\rmd v \\
    &=\int_{0}^{\infty}F(v;\bs{\theta}^{*})^{n-1}\rmd v-\int_{0}^{\infty}F(v;\bs{\theta}^{*})^{n}\rmd v. \\
\end{align}

\section{Detailed Calculation for Uniform Distribution} \label{app_uniform}
In this section, we omit the dependence on $t$ and denote $\bs{\theta}(t)=(x(t),y(t))$ and $\bs{\theta}^{*}(t)=(v_{\m}(t),v_{\M}(t))$ as $\bs{\theta}=(x,y)$ and $\bs{\theta}^{*}=(v_{\m},v_{\M})$ for simplicity. Suppose the uniform distribution of $\bs{\theta}^{*}=(v_{\m},v_{\M})$, and then the equilibrium bidding is calculated as
\begin{align}
    b^{*}(v)&=v-\frac{1}{F(v;\bs{\theta}^{*})^{n-1}}\int_{0}^{v}F(v';\bs{\theta}^{*})^{n-1}\rmd v' \\
    &=v-\left(\frac{v-v_{\m}}{v_{\M}-v_{\m}}\right)^{-(n-1)}\int_{v_{\m}}^{v}\left(\frac{v'-v_{\m}}{v_{\M}-v_{\m}}\right)^{n-1}\rmd v' \\
    &=v-\frac{1}{n}(v-v_{\m}) \\
    &=v_{\m}+\frac{n-1}{n}(v-v_{\m}).
\end{align}
Replacing $\bs{\theta}^{*}$ with $\bs{\theta}=(x,y)$, we obtain the bidding function as
\begin{align}
    b(v;\bs{\theta})=x+\frac{n-1}{n}(v-x).
\end{align}

The gradient of the payoff is calculated as
\begin{align}
    \left.\frac{\partial w_{\bs{\theta}^{*}}(\bs{\theta}',\bs{\theta})}{\partial \bs{\theta}'}\right|_{\bs{\theta}'=\bs{\theta}}&=-\int_{0}^{\infty}\left(\frac{f(v;\bs{\theta})}{F(v;\bs{\theta})}\right)^{-1}\left(\frac{f(v;\bs{\theta})}{F(v;\bs{\theta})}-\frac{f(v;\bs{\theta}^{*})}{F(v;\bs{\theta}^{*})}\right)f(v;\bs{\theta}^{*})F(v;\bs{\theta}^{*})^{n-1}\frac{\partial b(v;\bs{\theta})}{\partial\bs{\theta}}\rmd v \\
    &=-\int_{v_{\m}}^{v_{\M}}(v-x)\left(\frac{1}{v-x}-\frac{1}{v-v_{\m}}\right)\frac{1}{\Delta v}\left(\frac{v-v_{\m}}{\Delta v}\right)^{n-1}\frac{\partial b(v;\bs{\theta})}{\partial\bs{\theta}}\rmd v \\
    &=-\frac{x-v_{\m}}{(\Delta v)^{n}}\int_{v_{\m}}^{v_{\M}}(v-v_{\m})^{n-2}\frac{\partial b(v;\bs{\theta})}{\partial\bs{\theta}}\rmd v.
\end{align}
We further apply
\begin{align}
    \frac{\partial b(v;\bs{\theta})}{\partial x}=\frac{1}{n},\quad \frac{\partial b(v;\bs{\theta})}{\partial y}=0,
\end{align}
and obtain
\begin{align}
    \dot{x}&=\left.\frac{\partial w_{\bs{\theta}^{*}}(\bs{\theta}',\bs{\theta})}{\partial x'}\right|_{\bs{\theta}'=\bs{\theta}} \\
    &=-\frac{x-v_{\m}}{(\Delta v)^{n}}\int_{v_{\m}}^{v_{\M}}(v-v_{\m})^{n-2}\frac{\partial b(v;\bs{\theta})}{\partial x}\rmd v \\
    &=-\frac{1}{n}\frac{x-v_{\m}}{(\Delta v)^{n}}\int_{v_{\m}}^{v_{\M}}(v-v_{\m})^{n-2}\rmd v \\
    &=-\frac{x-v_{\m}}{n(n-1)\Delta v}, \\
    \dot{y}&=\left.\frac{\partial w_{\bs{\theta}^{*}}(\bs{\theta}',\bs{\theta})}{\partial y'}\right|_{\bs{\theta}'=\bs{\theta}} \\
    &=-\frac{x-v_{\m}}{(\Delta v)^{n}}\int_{v_{\m}}^{v_{\M}}(v-v_{\m})^{n-2}\frac{\partial b(v;\bs{\theta})}{\partial y}\rmd v \\
    &=0.
\end{align}

Finally, the payoffs in first-price and second-price auctions are
\begin{align}
    w^{\fst}_{\bs{\theta}^{*}}(\bs{\theta})=w_{\bs{\theta}^{*}}(\bs{\theta},\bs{\theta})&=\int_{0}^{\infty}(v-b(v;\bs{\theta}))f(v;\bs{\theta}^{*})F(v;\bs{\theta}^{*})^{n-1}\rmd v \\
    &=\int_{v_{\m}}^{v_{\M}}\frac{v-x}{n}\frac{1}{\Delta v}\left(\frac{v-v_{\m}}{\Delta v}\right)^{n-1}\rmd v \\
    &=\frac{1}{n(\Delta v)^{n}}\int_{v_{\m}}^{v_{\M}}(v-v_{\m})^{n}-(x-v_{\m})(v-v_{\m})^{n-1}\rmd v \\
    &=\frac{\Delta v}{n(n+1)}-\frac{x-v_{\m}}{n^{2}}, \\
    w^{\snd}_{\bs{\theta}^{*}}=w^{\fst}_{\bs{\theta}^{*}}(\bs{\theta}^{*})&=\frac{\Delta v}{n(n+1)}.
\end{align}

\section{Proof of Thm.~\ref{thm_inequivalence_1st}} \label{app_proof}
Without loss of generality, we can take $v_{\m}^{(1)}\le\cdots\le v_{\m}^{(K)}$, and it holds $v_{\m}^{(k)}<v_{\m}^{(k+1)}\Rightarrow \Delta v^{(k)}<\Delta v^{(k+1)}$ for all $k=1,\cdots,K-1$. For convenience, we omit below the description of the dependence on time $t$. Let $\Omega=[0,T]$ denote the whole interval of integration. Furthermore, we divide $\Omega$ into two subsets: $\Omega^{+}$ such that $\dot{x}\ge 0$ for $t\in\Omega^{+}$ and $\Omega^{-}$ such that $\dot{x}<0$ for $t\in\Omega^{-}$. Here, $\Omega=\Omega^{+}\cup\Omega^{-}$ trivially holds.

We also define $\tilde{k}(x)$ such that $v_{\m}^{(\tilde{k}(x))}\le x\le v_{\m}^{(\tilde{k}(x)+1)}$. For $t\in\Omega^{+}$, we obtain
\begin{align}
    \Delta v\dot{x}\rmd t\ge \Delta v^{(\tilde{k}(x)+1)}\rmd x,
    \label{ascent_distance}
\end{align}
meaning that when $x$ increases (i.e., $\dot{x}\ge 0\Leftrightarrow \rmd x\ge 0$), $\Delta v\dot{x}\rmd t$ is bounded from below in minimum $\Delta v$ such that $\dot{x}\propto -(x-v_{\m})>0\Leftrightarrow x<v_{\m}$. On the other hand, for $t\in\Omega^{-}$, we obtain
\begin{align}
    \Delta v\dot{x}\rmd t\ge \Delta v^{(\tilde{k}(x))}\rmd x,
    \label{descent_distance}
\end{align}
meaning that when $x$ decreases (i.e., $\dot{x}<0\Leftrightarrow \rmd x<0$), $\Delta v\dot{x}\rmd t$ is bounded from below in maximum $\Delta v$ such that $\dot{x}\propto -(x-v_{\m})\le 0\Leftrightarrow x\ge v_{\m}$.

Now, we obtain
\begin{align}
    \int_{\Omega}\Delta v(t)\dot{x}(t)\rmd t&=\int_{\Omega^{+}}\Delta v(t)\dot{x}(t)\rmd t+\int_{\Omega^{-}}\Delta v(t)\dot{x}(t)\rmd t
    \label{dvx_lbound_1}\\
    &\ge\int_{x(\Omega^{+})}\Delta v^{(\tilde{k}(x)+1)}\rmd x+\int_{x(\Omega^{-})}\Delta v^{(\tilde{k}(x))}\rmd x
    \label{dvx_lbound_2}\\
    &=\frac{1}{2}\left(\int_{x(\Omega^{+})}\Delta v^{(\tilde{k}(x)+1)}\rmd x-\int_{x(\Omega^{-})}\Delta v^{(\tilde{k}(x)+1)}\rmd x\right)+\frac{1}{2}\left(\int_{x(\Omega^{-})}\Delta v^{(\tilde{k}(x))}\rmd x-\int_{x(\Omega^{+})}\Delta v^{(\tilde{k}(x))}\rmd x\right)+O(1)
    \label{dvx_lbound_3}\\
    &=\frac{1}{2}\left(\int_{x(\Omega^{+})}(\Delta v^{(\tilde{k}(x)+1)}-\Delta v^{(\tilde{k}(x))})\rmd x-\int_{x(\Omega^{-})}(\Delta v^{(\tilde{k}(x)+1)}-\Delta v^{(\tilde{k}(x))})\rmd x\right)+O(1) \\
    &=\frac{1}{2}\left(\int_{x(\Omega^{+})}(\Delta v^{(\tilde{k}(x)+1)}-\Delta v^{(\tilde{k}(x))})|\rmd x|+\int_{x(\Omega^{-})}(\Delta v^{(\tilde{k}(x)+1)}-\Delta v^{(\tilde{k}(x))})|\rmd x|\right)+O(1) \\
    &\ge\frac{1}{2}\left(\min_{k}(\Delta v^{(k+1)}-\Delta v^{(k)})\right)\left(\int_{x(\Omega^{+})}|\rmd x|+\int_{x(\Omega^{-})}|\rmd x|\right)+O(1) \\
    &=\frac{1}{2}\left(\min_{k}(\Delta v^{(k+1)}-\Delta v^{(k)})\right)\int_{x(\Omega)}|\rmd x|+O(1).
    \label{dvx_lbound}
\end{align}
In Eq.~\eqref{dvx_lbound_1}, we divided the interval of integration, i.e., $\Omega=\Omega^{+}\cup\Omega^{-}$. In Eq.~\eqref{dvx_lbound_2}, we used Eqs.~\eqref{ascent_distance} and~\eqref{descent_distance}. In Eq.~\eqref{dvx_lbound_3}, we used Lem.~\ref{lem_equivalence} twice; For $h(x)=\Delta v^{(\tilde{k}(x)+1)}$, the lemma shows
\begin{align}
    \int_{x(\Omega^{+})}\Delta v^{(\tilde{k}(x)+1)}\rmd x=\frac{1}{2}\left(\int_{x(\Omega^{+})}\Delta v^{(\tilde{k}(x)+1)}\rmd x-\int_{x(\Omega^{-})}\Delta v^{(\tilde{k}(x)+1)}\rmd x\right)+O(1),
\end{align}
while for $h(x)=\Delta v^{(\tilde{k}(x))}$, the lemma also shows
\begin{align}
    \int_{x(\Omega^{-})}\Delta v^{(\tilde{k}(x))}\rmd x=\frac{1}{2}\left(\int_{x(\Omega^{-})}\Delta v^{(\tilde{k}(x))}\rmd x-\int_{x(\Omega^{+})}\Delta v^{(\tilde{k}(x))}\rmd x\right)+O(1).
\end{align}
Here, we remark that $h(x)$ is trivially bounded for both $h(x)=\Delta v^{(\tilde{k}(x)+1)}, \Delta v^{(\tilde{k}(x))}$.

Finally, the payoffs of first-price ($\bar{w}^{\fst}(\infty)$) and second-price ($\bar{w}^{\snd}(\infty)$) auctions are evaluated as
\begin{align}
    \bar{w}^{\fst}(\infty)-\bar{w}^{\snd}(\infty)&=\alpha\lim_{T\to\infty}\frac{1}{T}\int_{0}^{T}\Delta v \dot{x}\rmd t \\
    &\ge\alpha\lim_{T\to\infty}\frac{1}{T}\left\{\frac{1}{2}\left(\min_{k}(\Delta v^{(k+1)}-\Delta v^{(k)})\right)\int_{x(\Omega)}|\rmd x|+O(1)\right\} \\
    &=\frac{1}{2}\alpha\left(\min_{k}(\Delta v^{(k+1)}-\Delta v^{(k)})\right)\underbrace{\lim_{T\to\infty}\frac{1}{T}\int_{x(\Omega)}|\rmd x|}_{\rm time-average\ path\ length}.
    \label{bw1-bw2_lbound}
\end{align}
In conclusion, we obtain $\bar{w}^{\fst}(\infty)>\bar{w}^{\snd}(\infty)$.

\begin{lemma}[Equivalence between total ascent and descent]
\label{lem_equivalence}
We assume that $h(x)>0$ is an arbitrary function, but its integral, i.e., $h(x)$, is bounded. It is satisfied that
\begin{align}
    \int_{x(\Omega^{+})}h(x)\rmd x=-\int_{x(\Omega^{-})}h(x)\rmd x+O(1).
    \label{equivalence_ascent_descent}
\end{align}
\end{lemma}

\begin{proof}
First, we obtain
\begin{align}
    \int_{x(\Omega)}h(x)\rmd x=\left[h(x)\right]_{x(\Omega)}=H(x(T))-H(x(0))=O(1).
\end{align}
Using this, we prove
\begin{align}
    \int_{x(\Omega^{+})}h(x)\rmd x=\int_{x(\Omega)}h(x)\rmd x-\int_{x(\Omega^{-})}h(x)\rmd x=-\int_{x(\Omega^{-})}h(x)\rmd x+O(1).
\end{align}
By these equations, we obtain Eq.~\eqref{equivalence_ascent_descent}.

This lemma is intuitively interpreted as follows. For the time series of $x$ that oscillates sufficiently many times, the total distance of the ascent of $h$ (i.e., the LHS in Eq.~\eqref{equivalence_ascent_descent}) is almost equal to that of the descent of $h$ (i.e., the RHS in Eq.~\eqref{equivalence_ascent_descent}).
\end{proof}

\section{Example of Two-State Transition} \label{app_example}
This section is dedicated to strictly deriving $\bar{w}^{\fst}(T)-\bar{w}^{\snd}(T)$ as an exercise to follow the proof of Thm.~\ref{thm_inequivalence_1st}. For convenience, we omit below the description of the dependence on time $t$ again.

Let $\Omega=[0,T]$ denote the whole interval of integration. We divide this interval by $0=t_{0}^{-}\le t_{1}^{+}<t_{1}^{-}<\cdots<t_{L}^{+}\le t_{L}^{-}=T$ without the loss of generality. Here, we consider that $x$ takes its local maximum at the times of $t_l^{+}$, while it takes its local minimum at the times of $t_l^{-}$. Then, we can divide $\Omega$ into two subsets, $\Omega^{+}$ and $\Omega^{-}$, as
\begin{align}
    \Omega^{+}:=\bigcup_{l=1}^{L}[t_{l-1}^{-},t_{l}^{+}],\quad \Omega^{-}:=\bigcup_{l=1}^{L}(t_{l}^{+},t_{l}^{-}).
\end{align}

Next, for $t\in\Omega^{+}$, we obtain
\begin{align}
    \Delta v\dot{x}\rmd t=\Delta v^{(2)}\rmd x,
\end{align}
which corresponds to Eq.~\eqref{ascent_distance}. Here, however, note that by a special property in $K=2$, the equality strictly holds; when $x$ increases, $\bs{\theta}^{*}=\bs{\theta}^{(2)}$ holds since $v_{\m}^{(1)}<v_{\m}^{(2)}$ and $\dot{x}\propto -(x-v_{\m})$. For $t\in\Omega^{-}$, we also obtain
\begin{align}
    \Delta v\dot{x}\rmd t=\Delta v^{(1)}\rmd x,
\end{align}
which corresponds to Eq.~\eqref{descent_distance}.

Furthermore, we explain the abused notation, $x(\Omega^{+})$ and $x(\Omega^{+})$. Let $x_{l}^{+}:=x(t_{l}^{+})$ be the local maximum values and $x_{l}^{-}:=x(t_{l}^{-})$ be the local minimum values, and we can use more tractable expressions;
\begin{align}
    \int_{x(\Omega^{+})}\rmd x=\sum_{l=1}^{L}(x_{l}^{+}-x_{l-1}^{-}), \quad\int_{x(\Omega^{-})}\rmd x=\sum_{l=1}^{L}(x_{l}^{-}-x_{l}^{+}).
\end{align}
We also show Lem.~\ref{lem_equivalence} in more tractable expressions;
\begin{align}
    \sum_{l=1}^{L}(x_{l}^{+}-x_{l-1}^{-})&=\frac{1}{2}\left(\sum_{l=1}^{L}(x_{l}^{+}-x_{l-1}^{-})-\sum_{l=1}^{L}(x_{l}^{-}-x_{l}^{+})\right)+\frac{1}{2}(x_{l}^{-}-x_{0}^{-}), \\
    \sum_{l=1}^{L}(x_{l}^{-}-x_{l}^{+})&=\frac{1}{2}\left(\sum_{l=1}^{L}(x_{l}^{-}-x_{l}^{+})-\sum_{l=1}^{L}(x_{l}^{+}-x_{l-1}^{-})\right)+\frac{1}{2}(x_{l}^{-}-x_{0}^{-}),
\end{align}
where the last term is obviously $O(1)$.

To summarize the above, we obtain
\begin{align}
    \int_{\Omega}\Delta v\dot{x}\rmd t&=\int_{\Omega^{+}}\Delta v\dot{x}\rmd t+\int_{\Omega^{-}}\Delta v\dot{x}\rmd t& \\
    &=\Delta v^{(2)}\sum_{l=1}^{L}(x_{l}^{+}-x_{l-1}^{-})+\Delta v^{(1)}\sum_{l=1}^{L}(x_{l}^{-}-x_{l}^{+}) \\
    &=\Delta v^{(2)}\frac{1}{2}\left(\sum_{l=1}^{L}(x_{l}^{+}-x_{l-1}^{-})-\sum_{l=1}^{L}(x_{l}^{-}-x_{l}^{+})\right)+\Delta v^{(1)}\frac{1}{2}\left(\sum_{l=1}^{L}(x_{l}^{-}-x_{l}^{+})-\sum_{l=1}^{L}(x_{l}^{+}-x_{l-1}^{-})\right)+O(1) \\
    &=(\Delta v^{(2)}-\Delta v^{(1)})\frac{1}{2}\left(\sum_{l=1}^{L}|x_{l}^{+}-x_{l-1}^{-}|+\sum_{l=1}^{L}|x_{l}^{-}-x_{l}^{+}|\right)+O(1) \\
    &=\frac{1}{2}(\Delta v^{(2)}-\Delta v^{(1)})\int_{x(\Omega)}|\rmd x|+O(1).
\end{align}

In conclusion, we evaluate the time-average payoff as
\begin{align}
    \bar{w}^{\fst}(\infty)-\bar{w}^{\snd}(\infty)&=\alpha\lim_{T\to\infty}\frac{1}{T}\int_{0}^{T}\Delta v\dot{x}\rmd t \\
    &=\frac{1}{2}\alpha(\Delta v^{(2)}-\Delta v^{(1)})\underbrace{\lim_{T\to\infty}\frac{1}{T}\int_{x(\Omega)}|\rmd x|}_{\rm time-average\ path\ length}.
\end{align}
Notably, the difference in the time-average payoffs between first-price and second-price auctions is strictly given, as different from the lower bound in Eq.~\eqref{bw1-bw2_lbound}.

\section{Example of Non-Monotonic Cases} \label{app_necessity}
We used the monotonic relationship between $v_{\m}$ and $\Delta v$ in Thm.~\ref{thm_inequivalence_1st} and~\ref{thm_inequivalence_2nd}. This monotonicity might look like a too strong assumption, but it is necessary to obtain the direction of revenue inequivalence. To show the necessity of monotonicity, this section introduces a counterexample where the direction of revenue inequivalence cannot be determined only by the value distribution. Fig.~A\ref{FA01} shows the experiments for a counterexample given by the coupling of two minimum values ($v_{\m}^{(-)}=10$, $v_{\m}^{(+)}=20$) and two value intervals ($\Delta v^{(-)}=10$, $\Delta v^{(+)}=20$). In other words, $\bs{\theta}^{*}\in\{\bs{\theta}^{(-,-)},\bs{\theta}^{(-,+)},\bs{\theta}^{(+,-)},\bs{\theta}^{(+,+)}\}$ where $\bs{\theta}^{(-,-)}=(v_{\m}^{(-)},v_{\m}^{(-)}+\Delta v^{(-)})=(10,20)$, $\bs{\theta}^{(-,+)}=(v_{\m}^{(-)},v_{\m}^{(-)}+\Delta v^{(+)})=(10,30)$, $\bs{\theta}^{(+,-)}=(v_{\m}^{(+)},v_{\m}^{(+)}+\Delta v^{(-)})=(20,30)$, and $\bs{\theta}^{(+,+)}=(v_{\m}^{(+)},v_{\m}^{(+)}+\Delta v^{(+)})=(20,40)$, which does not satisfy the monotonicity obviously. In this example, we consider two cases of cyclic transitions: the cycle $\bs{\theta}^{(-,-)}\to \bs{\theta}^{(-,+)}\to \bs{\theta}^{(+,+)}\to \bs{\theta}^{(+,-)}\to \bs{\theta}^{(-,-)}\to \cdots$ in Fig.~A\ref{FA01}-A, whereas its reversed cycle $\bs{\theta}^{(-,-)}\to \bs{\theta}^{(+,-)}\to \bs{\theta}^{(+,+)}\to \bs{\theta}^{(-,+)}\to \bs{\theta}^{(-,-)}\to \cdots$ in Fig.~A\ref{FA01}-B. Interestingly, the direction of revenue equivalence is different between these two cases, even though both the possible states and the staying time at each state are the same.

One of our key findings is that $\bar{w}^{\fst}(\infty)-\bar{w}^{\snd}(\infty)$ is bounded by the travel distance $\rmd x$ in a moment (see Eqs.~\eqref{ascent_distance} and~\eqref{descent_distance}). This $\rmd x$ obviously depends on the non-equilibrium of $x$, i.e., $x-v_{\m}$. In other words, the value of $\Delta v$ immediately after $v_{\m}$ switched, where $x$ is highly non-equilibrium, is dominant in determining the direction of revenue inequivalence. Now focus on Fig.~A\ref{FA01}-A, in which $\Delta v$ takes $\Delta v^{(-)}$ immediately after $v_{\m}$ switches to $v_{\m}^{(-)}$, while $\Delta v$ takes $\Delta v^{(+)}$ immediately after $v_{\m}$ switches to $v_{\m}^{(+)}$. Thus, this can be approximately regarded as the transition between $\bs{\theta}^{(-,-)}$ and $\bs{\theta}^{(+,+)}$. Fig.~A\ref{FA01}-A is similar to Fig.~\ref{F02}-A. Next focus on Fig.~A\ref{FA01}-B, in which $\Delta v$ takes $\Delta v^{(+)}$ immediately after $v_{\m}$ switches to $v_{\m}^{(-)}$, while $\Delta v$ takes $\Delta v^{(-)}$ immediately after $v_{\m}$ switches to $v_{\m}^{(+)}$. Thus, this can be approximately regarded as the transition between $\bs{\theta}^{(-,+)}$ and $\bs{\theta}^{(+,-)}$. Fig.~A\ref{FA01}-B is similar to Fig.~\ref{F02}-C.

\begin{figure}[h!]
    \centering
    \includegraphics[width=0.8\hsize]{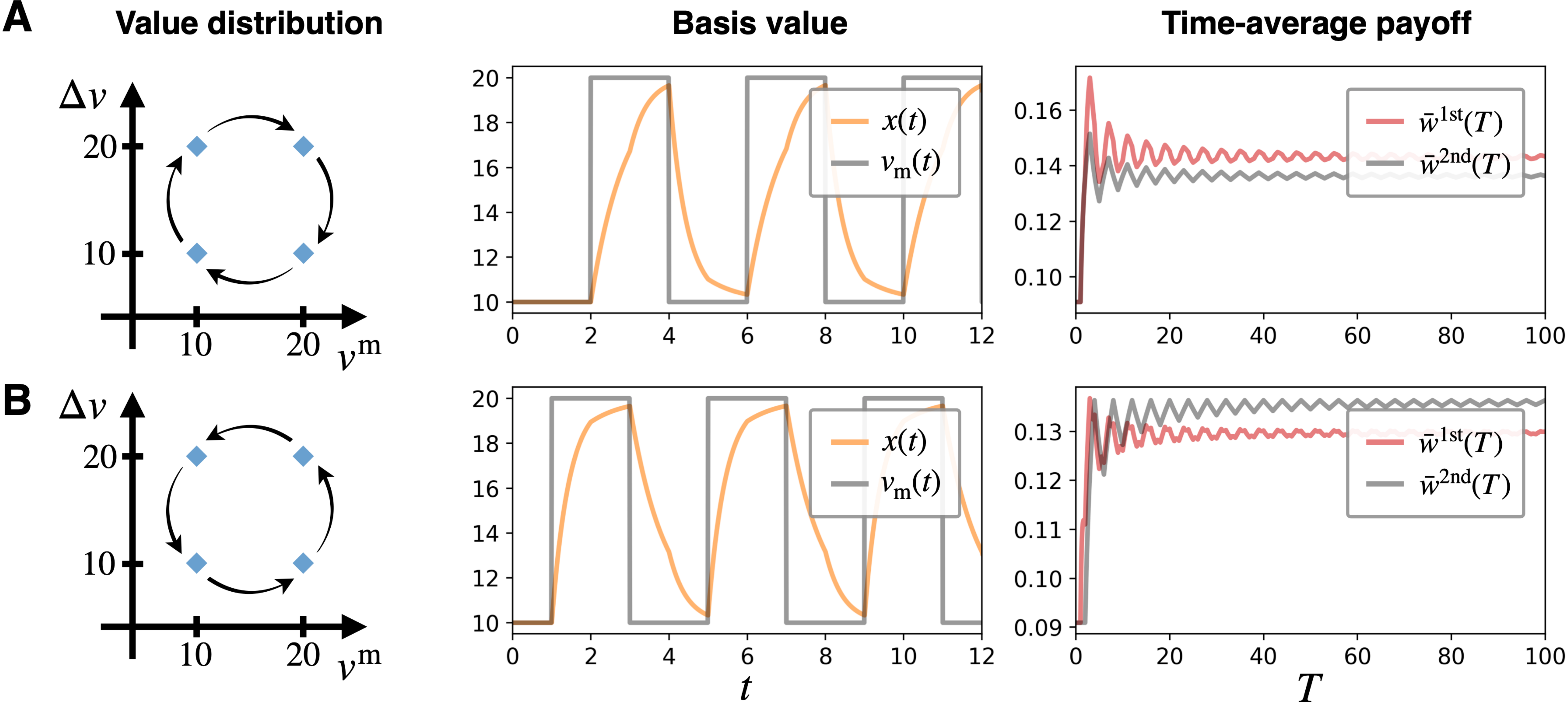}
    \caption{The experiments for a counterexample in which it cannot be determined whether bidders receive higher expected payoffs in the first-price auction than in the second-price auction. How to see each panel is the same as Fig.~\ref{F02}. The methods and parameters for all the experiments are the same as Fig.~\ref{F02}. Both (A) and (B) consider the same set of the value distributions $\bs{\theta}^{*}\in\{\bs{\theta}^{(-,-)},\bs{\theta}^{(-,+)},\bs{\theta}^{(+,-)},\bs{\theta}^{(+,+)}\}$, where we defined $\bs{\theta}^{(-,-)}:=(10,20)$, $\bs{\theta}^{(-,+)}:=(10,30)$, $\bs{\theta}^{(+,-)}:=(20,30)$, and $\bs{\theta}^{(+,+)}:=(20,40)$. (A). The cyclic transition of $\bs{\theta}^{(-,-)}\to \bs{\theta}^{(-,+)}\to \bs{\theta}^{(+,+)}\to \bs{\theta}^{(+,-)}\to \bs{\theta}^{(-,-)}\cdots$. (B). The reversed transition of $\bs{\theta}^{(-,-)}\to \bs{\theta}^{(+,-)}\to \bs{\theta}^{(+,+)}\to \bs{\theta}^{(-,+)}\to \bs{\theta}^{(-,-)}\cdots$.}
    \label{FA01}
\end{figure}


\end{document}